\documentclass{birkjour}
\usepackage{graphicx}
\usepackage{amsmath}
\usepackage{amssymb}
\usepackage{mathtools}
\usepackage{algorithm}
\usepackage{algorithmic}
\usepackage{tikz}
\usepackage[hyphens]{url}
\usepackage{hyperref}
\usepackage[hyphenbreaks]{breakurl}
\usetikzlibrary{arrows.meta, chains, positioning, decorations.pathreplacing, matrix, calc, fit}
 \newtheorem{theorem}{Theorem}[section]
 \newtheorem{corollary}[theorem]{Corollary}
 \newtheorem{lemma}[theorem]{Lemma}
 \newtheorem{proposition}[theorem]{Proposition}
 \theoremstyle{definition}
 \newtheorem{definition}[theorem]{Definition}
 \theoremstyle{remark}
 \newtheorem{remark}[theorem]{Remark}
 \newtheorem*{example}{Example}
 \numberwithin{equation}{section}
 
\def\NN{\mathbb N}
\def\ZZ{\mathbb Z}
\def\QQ{\mathbb Q}
\def\RR{\mathbb R}

\def\<#1>{\langle#1\rangle}

\DeclareMathOperator{\lc}{lc}

\DeclareMathOperator{\lm}{lm}

\DeclareMathOperator{\mon}{mon}
\DeclareMathOperator{\supp}{supp}
\DeclareMathOperator{\tail}{tail}
\DeclareMathOperator{\sdeg}{sdeg}

\floatname{algorithm}{Procedure}

\begin{document}

%
%
%
%
%
%
%
%
%

\title[Computing elements of certain form in ideals]{Computing elements of certain form in ideals to prove properties of operators}

\author[Clemens Hofstadler]{Clemens Hofstadler}

\address{%
Institute for Algebra\\
Johannes Kepler University Linz\\
Altenberger Stra\ss e 69\\
 4040 Linz, Austria}
\email{clemens.hofstadler@jku.at}

\thanks{Supported by the Austrian Science Fund (FWF):  P~31952 and P~32301}

\author[Clemens G.\ Raab]{Clemens G.\ Raab}
\address{%
Institute for Algebra\\
Johannes Kepler University Linz\\
Altenberger Stra\ss e 69\\
 4040 Linz, Austria}
\email{clemensr@algebra.uni-linz.ac.at}

\author[Georg Regensburger]{Georg Regensburger}
\address{%
Institute of Mathematics\\
University of Kassel\\
Heinrich-Plett-Straße 40\\
34132 Kassel, Germany}
\email{regensburger@mathematik.uni-kassel.de}



\date{\today}

\begin{abstract}

Proving statements about linear operators expressed in terms of identities often leads to finding elements of certain form in noncommutative polynomial ideals.
We illustrate this by examples coming from actual operator statements and discuss relevant algorithmic methods for finding such polynomials based on noncommutative Gr\"obner bases.  
In particular, we present algorithms for computing the intersection of a two-sided ideal with a one-sided ideal as well as for computing homogeneous polynomials in two-sided ideals and monomials in one-sided ideals. 
All methods presented in this work are implemented in the \textsc{Mathematica} package \texttt{OperatorGB}.
 
\end{abstract}

\maketitle
\tableofcontents

\section{Introduction}

The main motivation for the work presented here is to automatize proofs of statements about matrices and linear operators.
Many properties of matrices and linear operators can be expressed in terms of identities.
Algebraically, these identities can be modeled as free noncommutative polynomials whose indeterminates represent the basic operators involved.
From the polynomials corresponding to known or assumed identities of these operators a (two-sided) ideal is generated.
The most basic instance of a computational problem that can be used for proving a statement about linear operators consists in identifying a given polynomial as a member of a given ideal. 
More generally, proving a statement about linear operators translates into finding polynomials of a certain form in the given ideal.
In this paper, we discuss well-known and present novel algorithmic methods based on noncommutative Gr\"obner basis computations to do this and illustrate them by examples.

 Gr\"{o}bner bases of ideals arising from matrix identities have already been used in the pioneering work \cite{HeltonWavrik1994,HeltonStankusWavrik1998}.
 Recently, proving operator identities using Gr\"obner basis computations and related questions were also addressed in \cite{LevandovskyySchmitz2020}.
However, not every computation with noncommutative polynomials can be translated into a computation with operators, since addition and composition of operators are restricted by their domains and codomains.
In \cite{RaabRegensburgerHosseinPoor2021}, a general algebraic framework was developed that provides conditions on the generators of the ideal and the given polynomial to ensure that ideal membership of that polynomial is equivalent to the existence of a valid computation with actual operators proving the corresponding identity. 
More precisely, in this way a single computation with polynomials verifying ideal membership proves the statement about operators in an abstract setting in such a way that analogous statements in various concrete settings (e.g.\ for matrices, linear bounded operators, $C^*$-algebras, \ldots) follow immediately in a rigorous way.
To encode domains and codomains of operators, the framework uses a quiver, where vertices correspond to spaces, edges correspond to basic operators mapping between those spaces, and the labels of edges are indeterminates in a noncommutative polynomial ring. The central notion is then given by polynomials that are \emph{compatible with the quiver}, i.e., only consist of monomials that are labels of paths between the same pair of vertices in the quiver.

Sometimes identities arising from the properties of operators contain undetermined operators.
One class of such properties are conditions on ranges and kernels, like the inclusion of ranges $\mathcal{R}(A)\subseteq \mathcal{R}(B)$ of operators $A,B$. 
For linear operators over a field, such a range inclusion can be translated into the existence of a factorization $A=BX$ for some operator $X$.
In our setting, if we want to prove such a range inclusion, we have to find an element in the ideal generated by the assumptions that represents an identity of the form $A=BX$.
More generally, solving operator equations explicitly leads to searching for elements of a particular form in a given ideal.
This problem also arises when we want to use properties like injectivity or surjectivity of some operators in proofs. For example, to apply injectivity of an operator $A$, we have to search for elements representing an identity of the form $AB=AC$ to conclude that $B=C$. 
Also other properties can be described by the cancellability of certain left or right factors.

To motivate our work and illustrate our methods, we look at two concrete examples that deal with the solvability of operator equations.
First, we consider the following characterization of the solvability of the operator equation $AXB=C$ in Hilbert spaces from Proposition~3.3 in \cite{Arias}.
In the following, $P^\dagger$ denotes the Moore-Penrose inverse of a bounded linear operator $P$ and $P^*$ the adjoint operator.

\begin{theorem}\label{thm example 1}
Let $A\colon \mathcal{H}_4 \to \mathcal{H}_2$, $B\colon \mathcal{H}_1 \to \mathcal{H}_3$, $C\colon \mathcal{H}_1 \to \mathcal{H}_2 $ be bounded linear operators on complex Hilbert spaces.
There exists a bounded linear operator $X\colon \mathcal{H}_3 \to \mathcal{H}_4$ such that $AXB = C$ if and only if $\mathcal{R}(C) \subseteq \mathcal{R}(A)$ and $\mathcal{R}((A^\dagger C)^*) \subseteq \mathcal{R}(B^*)$.
\end{theorem}

We discuss how to prove the sufficiency of the range inclusions in Theorem~\ref{thm example 1} by a computation with polynomials.
First, all properties appearing in the statement have to be phrased in terms of identities.
For our assumptions, the postulated range inclusions can be translated into the existence of operators $Y,Z$ such that $C = AY$ and $(A^\dagger C)^* = B^* Z$.
The existence of the Moore-Penrose inverse $A^\dagger$ is encoded in terms of identities by adding the four defining equations of $A^\dagger$ to our assumptions.
These identities, the two identities for the range inclusions, and the respective adjoint statements can be translated into a set $F$ of 10 noncommutative polynomials with integer coefficients in 12 indeterminates.
The quiver encoding the domains and codomains of the operators involved consists of 4 vertices and 12 edges.
It is depicted in Figure~\ref{fig:quiver}.
Then, finding a solution $X$ of $AXB = C$ corresponds to finding a compatible polynomial $axb - c \in (F)$, where $a,b,c$ are known but $x$ is an unknown polynomial, in the ideal generated by $F$.
In Section~\ref{sec ideal intersection}, we discuss how this can be done using ideal intersections.
In particular, we discuss well\nobreakdash-known ideal intersection techniques~\cite{Nor98,Mora2016} and we present a new algorithm to compute generators of the intersection of an ideal with a one\nobreakdash-sided ideal, which will allow us to solve this problem.
We also characterize when an ideal has a finite right generating set.

However, not all properties of operators ultimately lead to statements about polynomials that can be solved by ideal intersections. 
One example of such a property is determining whether a given operator is positive.
Recall that a linear operator $P$ is called positive if there exists another linear operator $Q$ such that $P$ factors as $P = Q^* Q$.
A concrete statement where one has to determine the positivity of a given operator is the following part of Theorem~4.3 in~\cite{Arias}.

\begin{theorem}\label{thm example 2}
Let $A\colon \mathcal{H}_3 \to \mathcal{H}_2$, $B\colon \mathcal{H}_1 \to \mathcal{H}_3$, $C\colon \mathcal{H}_1 \to \mathcal{H}_2 $ be bounded linear operators on complex Hilbert spaces such that $\mathcal{R}(B) \subseteq \mathcal{R}(A^*)$.
If there exists a bounded, positive linear operator $X\colon \mathcal{H}_3 \to \mathcal{H}_3$ such that $AXB = C$, then $B^* A^\dagger C$ is positive.
\end{theorem}

Proving this statement by a computation with polynomials ultimately leads to the following problem.
Given the ideal $(F)$ generated by the assumptions $F$ of the theorem, find a polynomial of the form $p - q^* q \in (F)$, where $p$ is known but the polynomial $q$, and therefore also $q^*$, are unknown.
Here, $q^*$ denotes the image of $q$ under the involutive antiautomorphism ${}^*$ that sends each variable $x_i$ to an adjoint variable $x_i^*$ and each $x_i^*$ back to $x_i$.
Ideal intersections are not useful here since the second term $q^* q$ is completely unknown.
To nevertheless solve this problem if $q$ is a noncommutative monomial, we have generalized a method to compute monomials in a commutative ideal~\cite{SST00, Mil16}.
In Section~\ref{sec homogeneous part}, we present this new algorithm, that allows to compute homogeneous polynomials in noncommutative ideals, and discuss how it can be used to prove Theorem~\ref{thm example 2}.
Additionally, in Section~\ref{sec monomial part} we show how a simple adaptation of this procedure allows to find monomials in one-sided ideals.

As a consequence of the undecidability of the word problem, also ideal membership in rings of noncommutative polynomials is undecidable.
This is, among other things, also reflected in the fact that not all finitely generated ideals in such rings have a finite Gr\"obner basis.
Nevertheless, an analog of Buchberger's algorithm can be used to enumerate a (possibly infinite) Gr\"{o}bner basis \cite{Mora1994}.
Since all the methods presented in this work rely on noncommutative Gr\"obner basis computations, they are in fact not terminating algorithms but rather enumeration procedures.
In practice, however, these tools still often allow to verify ideal membership for concrete polynomials.
More generally, for finding polynomials of a certain form in an ideal, we first specify a certain subideal (e.g.\ by ideal intersections), which has to contain all polynomials of the desired form.
Then, we enumerate generators of this subideal and search for a polynomial of the desired form.
Note that, in general, there is no guarantee that we can find such a polynomial among the generators of the subideal.
However, in practice, this was the case in almost all algebraic proofs of operator statements we considered so far.

All the procedures presented here are implemented in the \textsc{Mathematica} software package \texttt{OperatorGB}~\cite{Hof20},
which also provides support for proving properties of matrices and operators along these lines using the framework mentioned above.
The package is available at \url{https://clemenshofstadler.com/software/} along with a notebook and PDF file containing detailed proofs of the examples discussed in this paper.

Finally, we note that the methods presented in this work are not exhaustive.
There exist other procedures that also allow to compute elements of certain form in noncommutative polynomial ideals.
For example, certain problems can also be solved by computing the intersection of a (two-sided) ideal with a subalgebra of the ring of noncommutative polynomials \cite{Nor98}.
Furthermore, the problems that Theorem~\ref{thm example 1} and~\ref{thm example 2} lead to, could also be considered as factorization problems in a suitable quotient of the ring of noncommutative polynomials.
In~\cite{BHL17}, it was shown that this quotient is a finite factorisation domain if it admits a finite dimensional filtration and a terminating algorithm was given that allows to compute all such factorisations in the affirmative case.
Unfortunately, however, for the problems arising from statements about operators, we usually do not have access to such a finite dimensional filtration (in fact, most of the time we do not even know whether it exists) as this would
require a description of a complete Gr\"obner basis.

To keep the presentation self-contained, we briefly recall the most important notions about noncommutative polynomials and the algebraic proof framework from \cite{RaabRegensburgerHosseinPoor2021} in the following section.


\section{Preliminaries and notation}

We denote by $\<X>$ the free monoid over $X = \{x_1,\dots,x_n\}$ and we let $K\<X>$ be the free algebra generated by $X$ over a field $K$.
When speaking about an algebra over $K$, or a $K$-algebra for short, we always mean an (associative) $K$-algebra with unit element. 
We consider the elements in $K\<X>$ as \emph{noncommutative polynomials} with coefficients in $K$ and monomials in $\<X>$.
In $K\<X>$, indeterminates still commute with coefficients but not with each other.

For a monomial $m = x_{i_1}\dots x_{i_k} \in \<X>$, the quantity $k$ is called the \emph{length} of $m$ and denoted by $|m|$.
A  \emph{multiple} of $m$ is any monomial $m'$ of the form $m' = a m b$ with $a,b\in \<X>$.
If $a = 1$, then $m'$ is called a  \emph{right multiple} of $m$.
If $m'$ is a multiple of $m$, then $m$ is said to \emph{divide} $m'$.
Furthermore, $m$ is called  \emph{(right) reducible} modulo $M \subseteq \<X>$ if $m$ is a (right) multiple of some monomial in $M$, otherwise it is called \emph{(right) irreducible} modulo $M$.
We fix a monomial ordering $\preceq$ on $\<X>$, i.e., $\preceq$ is a well-ordering on $\<X>$ such that $m \preceq m'$ implies $amb \preceq am'b$ for all $a,b,m,m' \in \<X>$.

For a given set of polynomials $F \subseteq K\<X>$, we denote by $(F)$ and $(F)_\rho$ the (two-sided) ideal, respectively the right ideal, generated by $F$.
When working with one-sided ideals, we restrict ourselves to right ideals, since the situation for left ideals is completely symmetric and all theorems about right ideals also hold, mutatis mutandis, for left ideals.
In order to help with the distinction between two-sided ideals and right ideals, we denote the former by capital letters, e.g.\ $I,J,\dots$,
and the latter by capital letters with an additional subscript $\rho$, e.g.\ $I_\rho, J_\rho, \dots$.

Furthermore, we let $\lm(F) = \{\lm(f) \mid 0 \neq f \in F\}$, where $\lm(f)$ denotes the leading monomial of $f$ w.r.t.\ $\preceq$.
The set $F$ is called  \emph{(right) reduced} if every element in $F$ is monic and if, for all $f \in F$, all monomials appearing in $f$ are (right) irreducible modulo $\lm(F \setminus \{f\})$.
We note that a finite set $F$ can always be transformed into a finite (right) reduced set $F'$ such that $(F) = (F')$ (resp.~s.t.~$(F)_\rho = (F')_\rho$).
See for example \cite{Xiu12} on how this can be done.

Given an ideal $I \subseteq K\<X>$, we say that $G \subseteq K\<X>$ is a  \emph{(two-sided) Gr\"obner basis} of $I$ if $(G) = I$ and $(\lm(G)) = (\lm(I))$.
For a right ideal $I_\rho \subseteq K\<X>$, a right generating set $G \subseteq K\<X>$ is called a  \emph{right Gr\"obner basis} of $I_\rho$ if $(\lm(G))_\rho = (\lm(I_\rho))_\rho$.
If $G$ is (right) reduced, it is called the  \emph{reduced (right) Gr\"obner basis}.
We note that the reduced (right) Gr\"obner basis of an ideal (resp.\ a right ideal) is unique.
Furthermore, the reduced right Gr\"obner basis of a right ideal $I_\rho$ can be computed by right reducing any generating set of $I_\rho$.
The following lemma expresses this fact and implies that every finitely generated right ideal has a finite Gr\"obner basis; no matter the chosen monomial ordering.

\begin{lemma}\label{lemma right gb}
Let $I_\rho \subseteq K\<X>$ be a right ideal with right generating set $G \subseteq K\<X>$.
If $G$ is right reduced, then $G$ is the reduced right Gr\"obner basis of $I_\rho$.
\end{lemma}

\begin{proof}
See~\cite[Proposition 4.4.4]{Xiu12}.
\end{proof}

All the techniques presented in this work rely on the elimination property of noncommutative Gr\"obner bases \cite{BB98}.
If we let $I \subseteq K\<X>$ be an ideal and $Y \subseteq X$, then, analogous to the commutative case,
the elimination property of noncommutative Gr\"obner bases allows to easily obtain a Gr\"obner basis of the elimination ideal $I \cap K\<X \setminus Y>$  
from a suitable Gr\"obner basis of $I$.
A Gr\"obner basis is suitable to do this, if it is computed w.r.t.\ a special monomial ordering, called an \emph{elimination ordering}.
Such an ordering is defined as follows.
A monomial ordering is called an \emph{elimination ordering} for $Y \subseteq X$ if $\lm(f) \in \< X \setminus Y>$ implies $f \in K \< X \setminus Y>$ for all $f \in K\<X>$.
The following theorem states the elimination property of noncommutative Gr\"obner bases.

\begin{theorem}\label{thm elim}
Let $I \subseteq K\<X>$ be an ideal and $Y \subseteq X$. 
Furthermore, let $G$ be a Gr\"obner basis of $I$ computed w.r.t.\ an elimination ordering for $Y$.
Then, $G \cap  K\<X \setminus Y>$ is a Gr\"obner basis of the elimination ideal $I \cap K\<X \setminus Y>$.
\end{theorem}

\begin{proof}
See~\cite[Theorem~3.2]{BB98}. 
\end{proof}

We note that Theorem~\ref{thm elim} also applies in the same way to right ideals. 
For further information on Gr\"obner bases in the free algebra, see for example~\cite{Mora1994,Xiu12,Mora2016,Hof20}, and also~\cite{letterplace} for an overview on available software and further references.

To end this section, we recall some of the notions related to the framework developed in \cite{RaabRegensburgerHosseinPoor2021} that are also relevant for the work presented here.
For a short, more formal summary see also the appendix in~\cite{Hartwig}.
In order to translate identities of linear operators into noncommutative polynomials, each basic operator is replaced uniformly by a unique noncommutative indeterminate from some set $X$ and the difference of the left and right hand side of each identity is formed.
In this process, all information about the different domains and codomains of the operators is lost.
In order to preserve this information, it is encoded in a \emph{labelled quiver} $Q$.
This is a directed multigraph where edges are labelled by the indeterminates in $X$.
More precisely, the domains and codomains of the basic operators correspond to the vertices of $Q$ and an edge with label $x \in X$ is drawn from a vertex $v$ to a vertex $w$ if and only if the operator that has been replaced by
$x$ maps from the space represented by $v$ to the space represented by $w$. 
With this, a concatenation of paths in $Q$ corresponds to a multiplication of monomials in $\<X>$.
Then, a polynomial in $K\<X>$ is called \emph{compatible} with $Q$ if all its monomials correspond to paths in $Q$ which all have the same start and end.
Informally speaking, a polynomial is compatible with $Q$ if and only if it can be translated into a valid operator in the setting encoded by $Q$.
The main result of \cite{RaabRegensburgerHosseinPoor2021} then says the following. If the polynomials $F \subseteq K\<X>$ representing the assumptions of a statement about operators
as well as the polynomial $f \in K\<X>$ corresponding to the claimed property are all compatible with the quiver $Q$ encoding the domains and codomains,
then the ideal membership $f \in (F)$ is equivalent to the correctness of the statement about operators.


\section{Ideal intersections}\label{sec ideal intersection}

In the introduction, we have already indicated that ideal intersections provide a useful tool for finding elements of certain forms in an ideal $I \subseteq K\<X>$.
For example, they allow us to systematically search for elements of the form $l a r \in I$, where $a \in K\<X>$ is known but $l,r \in K\<X>$ are unknown.
This can be done by intersecting $I$ with the ideal $(a)$.
Moreover, we cannot only find polynomials where only a certain factor in the middle is known but also elements where, for example, only a specific prefix is given.
More precisely, intersecting the ideal $I$ with the right ideal $(a)_\rho$ allows to compute elements of the form $a r \in I$.

To support the proof of Theorem~\ref{thm example 1} by computing ideal intersections, we can use the following approach.
To find a suitable polynomial $axb - c \in (F)$ in the ideal generated by the assumptions $F$ of the statement, 
we intersect $(F)$ with the right ideal $(a,c)_\rho$ and enumerate generators of this intersection.
Among these generators, we then search for a polynomial of the form $a r - c$ with the additional requirement that the free part $r$ has to end with the variable $b$.
If we can find such a polynomial and show that it is compatible with the quiver $Q$ encoding the domains and codomains of the statement, then this proves Theorem~\ref{thm example 1}.
This approach is carried out successfully in the example at the end of this section.

In this section, we discuss how generating sets of several different kinds of ideal intersections can be computed. 
Recall that $I,J,\dots$ denote two\nobreakdash-sided ideals while $I_\rho, J_\rho,\dots$ denote right ideals.

In the case of commutative polynomials, two ideals $I,J \subseteq K[X]$ can be intersected by computing the elimination ideal $tI + (1-t)J \cap K[X]$, where $t$ denotes a new tag variable.
It is well-known that the same also works for intersecting one-sided ideals in the free algebra (see e.g.~\cite[Remark~48.8.6]{Mora2016}).

\begin{theorem}\label{thm intersect right ideals}
Let $I_\rho = (f_1,\dots,f_r)_\rho$, $J_\rho = (g_1,\dots,g_s)_\rho \subseteq K\<X>$ be two finitely generated right ideals.
Consider the right ideal 
\begin{align*}
	H_\rho = \left(t f_i, (1-t) g_j \mid 1 \leq i \leq r, 1 \leq j \leq s\right)_\rho.
\end{align*}
Then, $I_\rho \cap J_\rho = H_\rho \cap K\<X>$.
\end{theorem}

Consequently, making use of the elimination property of noncommutative Gr\"obner bases, 
a right Gr\"obner basis of $I_\rho \cap J_\rho$ can be computed by removing all polynomial involving the variable $t$ from a right Gr\"obner basis of $H_\rho$ that is computed w.r.t.\ an elimination ordering for $\{t\}$.
Lemma~\ref{lemma right gb} ensures that the Gr\"obner basis computed in this way is always finite.

In~\cite{Nor98}, it was shown that the elimination property of noncommutative Gr\"obner bases also allows to compute the intersection of two-sided ideals provided that we introduce the commutator relations
$tx_k - x_kt$, for all $k = 1,\dots, n$, between the new tag variable $t$ and the elements of $X$. 

\begin{theorem}
Let $I = (f_1,\dots,f_r)$, $J = (g_1,\dots,g_s) \subseteq K\<X>$ be two finitely generated ideals.
Consider the ideal 
\begin{align*}
	H = \left(t f_i, (1-t) g_j, tx_k - x_k t \mid 1 \leq i \leq r, 1 \leq j \leq s, 1 \leq k \leq n\right).
\end{align*}
Then, $I \cap J = H \cap K\<X>$.
\end{theorem}

\begin{proof}
See~\cite[Theorem~2]{Nor98}.
\end{proof}

Hence, as for right ideals, a Gr\"obner basis of the intersection can be computed by computing a Gr\"obner basis of $H$ w.r.t.\ an elimination ordering for $\{t\}$
and by intersecting this Gr\"obner basis with $K\<X>$.
However, in contrast to the previous case, we cannot expect this Gr\"obner basis to be finite in general.

Next, we consider the third and for our application most relevant type of ideal intersection  -- the intersection of a two-sided ideal $I \subseteq K\<X>$ with a right ideal $J_\rho \subseteq K\<X>$.
The basic idea to compute the intersection $I \cap J_\rho$ is to consider $I$ as a right ideal and to intersect two right ideals.
As Theorem~\ref{thm intersect right ideals} provides an algorithmic way to compute the intersection of two right ideals, it remains to discuss how to obtain a right generating set of an ideal $I$.
The following proposition is a spezialisation of a result in~\cite{Gre00} for path algebras and tells us how the reduced two-sided Gr\"obner basis of an ideal $I$ can be used to obtain a right Gr\"obner basis of $I$.
In particular, this provides a way to obtain a right generating set of a two-sided ideal.
We recall that a monomial $p \in \<X>$ is called a \emph{proper prefix} of a monomial $m \in \<X>$ if $m = p q$ with $1 \neq q \in \<X>$.

\begin{proposition}\label{prop right GB}
Let $I \subseteq K\<X>$ be an ideal and let $G$ be its reduced Gr\"obner basis.
Then, the set
\begin{align*}
	\rho(I) \coloneqq \{&mg \mid m\in \<X>, g\in G,\\
		&p \text{ is irreducible modulo }\lm(G)\text{ for any proper prefix }p\text{ of }\lm(mg)\}
\end{align*}
is a right Gr\"obner basis of $I$.
Furthermore, $\lm(f)$ is right irreducible modulo $\lm(\rho(I) \setminus \{f\})$ for all $f \in \rho(I)$.
\end{proposition}

\begin{proof}
Follows from~\cite[Proposition~7.1]{Gre00} since $K\<X>$ can be viewed as a path algebra with only one vertex.
\end{proof}

\begin{remark}
Equivalently, we can also write $\rho(I)$ as 
 \begin{align}\label{equation right gb}
 \begin{aligned}
	\rho(I) = \{&mg \mid m\in B_I, g\in G, \\ 
		&p \in B_I \text{ for any proper prefix }p\text{ of } \lm(mg)\},
\end{aligned}
\end{align}
where $G$ is the reduced Gr\"obner basis of $I$ and $B_I = \<X> \setminus \lm(I)$ is the set of all monomials which are irreducible modulo $\lm(I)$.
Note that the equivalence classes of the monomials in $B_I$ form a $K$-vector space basis of $K\<X>/I$.
\end{remark}

It follows from the uniqueness of the reduced Gr\"obner basis, that, for a fixed monomial ordering, the set $\rho(I)$ only depends on $I$.
In the following, we discuss under which conditions a two-sided ideal $I$ is finitely generated as a right ideal.
The following result tells us that this is the case if and only if the set $\rho(I)$ is finite.

\begin{proposition}\label{prop equiv finite right GB}
Let $I \subseteq K\<X>$ be an ideal. 
Then, the following are equivalent.
\begin{enumerate}
	\item $I$ is finitely generated as a right ideal.
	\item $I$ has a finite right Gr\"obner basis.
	\item The set $\rho(I)$ is finite.
\end{enumerate}
\end{proposition}

\begin{proof}
The implications $(3) \Rightarrow (2) \Rightarrow (1)$ are clear. 
Furthermore, the implication $(1) \Rightarrow (2)$ follows immediately from Lemma~\ref{lemma right gb} and the fact that a finite set of polynomials can always be transformed into a finite right reduced set.
For the implication $(2) \Rightarrow (3)$, assume that $I$ has a finite right Gr\"obner basis, say $G = \{g_1,\dots,g_m\} \subseteq K\<X>$, and assume that $\rho(I)$ is infinite.
Since $\rho(I)$ is a right Gr\"obner basis of $I$, there exist $g_1',\dots,g_m'\in \rho(I$) and $b_1,\dots,b_m\in\<X>$ such that $\lm(g_i) = \lm(g_i')b_i$ for all $ i =1,\dots,m$.
Now, let $g \in \rho(I) \setminus \{g_1',\dots,g_m'\}$ be arbitrary but fixed. 
We note that such an element exists since $\rho(I)$ is infinite.
Then, there exist $1 \leq i \leq m$ and $b \in \<X>$ such that $\lm(g) = \lm(g_i)b = \lm(g_i') b_i b$.
Since $g \neq g_i'$, this is a contradiction to the assertion of Proposition~\ref{prop right GB} that $\lm(f)$ is right irreducible modulo $\lm(\rho(I) \setminus \{f\})$ for all $f \in \rho(I)$.
\end{proof}

The set $\rho(I)$ depends on the monomial ordering w.r.t.\ which the reduced Gr\"obner basis of $I$ is computed.
However, Lemma~\ref{lemma right gb} and Proposition~\ref{prop equiv finite right GB} imply that the finiteness of $\rho(I)$ is independent of the monomial ordering.

We will now investigate under which conditions the set $\rho(I)$ is finite.
Clearly, if the reduced Gr\"obner basis $G$ of $I$ is infinite, then so must be $\rho(I)$, or equivalently, the finiteness of $\rho(I)$ implies the finiteness of $G$. 
The converse, however, is not universally true. 
The finiteness of $G$ alone is only sufficient to guarantee that also $\rho(I)$ is finite if $I$ is the trivial ideal $I = \{0\}$.
In this case, $G = \emptyset = \rho(I)$.
We will see that in all other cases, we need the additional requirement that also $B_I$ is finite.
To this end, we recall that an ideal $I\subseteq K\<X>$ is called \emph{zero-dimensional} if the quotient algebra $K\<X>/I$ is finite dimensional as a $K$-vector space,
or equivalently, if the set $B_I$ is finite.
It is well-known that any zero-dimensional ideal has a finite Gr\"obner basis w.r.t.\ any monomial ordering (see e.g.~5.1 in~\cite{Mor87}).

\begin{lemma}\label{lemma finite two sided GB}
Let $I \subseteq K\<X>$ be a zero-dimensional ideal.
Then, the reduced Gr\"obner basis of $I$ is finite. 
\end{lemma}

In fact, this condition is also sufficient for an ideal to have a finite right Gr\"obner basis.

\begin{lemma}\label{lemma finite right GB suff}
Let $I \subseteq K\<X>$ be a zero-dimensional ideal.
Then, the set $\rho(I)$ is finite.
\end{lemma}

\begin{proof}
It follows from Lemma~\ref{lemma finite two sided GB} that the reduced Gr\"obner basis of $I$ is finite.
Then, the finiteness of $\rho(I)$ follows directly from the representation \eqref{equation right gb} and the fact that $B_I$ is finite.
\end{proof}

Surprisingly, for a non-trivial ideal $I$, the property of $I$ being zero-dimensional is also a necessary condition for it to have a finite right Gr\"obner basis.
This fact is captured in the following lemma.

\begin{lemma}\label{lemma finite right GB nec}
Let $I \subseteq K\<X>$ be a non-trivial ideal.
If $I$ has a finite right Gr\"obner basis, then $I$ is zero-dimensional.
\end{lemma}

\begin{proof}
Let $G$ be a finite right Gr\"obner basis of $I$ and denote $N = \max \{ |\lm(g)| \mid g \in G\}$.
Let $m \in \<X>$ such that $|m| \geq N$ and choose $g \in G$ arbitrary.
Then, $m g \in I$ and therefore, since $G$ is a right Gr\"obner basis, there exists $g' \in G$ such that $m \lm(g)$ is a right multiple of $\lm(g')$.
By the way $m$ was chosen, $m$ in fact has to be a right multiple of $\lm(g')$.
Consequently, $m \notin B_I$.
Since $m$ was arbitrary of length at least $N$, this shows that $B_I$ can contain only finitely many elements.
Hence, $I$ is zero-dimensional.
\end{proof}

Combining Proposition~\ref{prop equiv finite right GB} with Lemma~\ref{lemma finite right GB suff} and Lemma~\ref{lemma finite right GB nec}, we get the following theorem.

\begin{theorem}
Let $I \subseteq K\<X>$ be an ideal. 
Then, the following are equivalent.
\begin{enumerate}
	\item $I$ is finitely generated as a right ideal.
	\item $I$ has a finite right Gr\"obner basis.
	\item The set $\rho(I)$ is finite.
	\item  $I = \{0\}$ or $I$ is zero-dimensional.
\end{enumerate}
\end{theorem}

However, the condition of $I$ being zero-dimensional is very strong and in practice usually not fulfilled. 
Typically, the reduced Gr\"obner basis of $I$ is not even finite.
In such cases, there is no chance to obtain a finite right generating set.
Consequently, we have to content ourselves with finite approximations of $\rho(I)$ and can therefore only work with a right subideal $I_\rho' \subsetneq I$.
A priori, it is usually not clear how to choose such a finite approximation.
For the operator statements we looked at so far, we simply selected all elements in $\rho(I)$ up to a certain degree bound.
Thus far, this has worked very well, allowing us to find the correct polynomials.

Moreover, in case of our applications, we can choose the subideal $I_\rho'$ so that it still contains all polynomials that are of interest to us and such that $I_\rho'$ is more likely to have a finite right generating set. 
This follows from the fact that for proving statements about operators, we are only interested in polynomials in $I$ that are compatible with a certain quiver $Q$.  
Hence, if the objective is to find polynomials in $I$ that are compatible with a given quiver $Q$, then $\rho(I)$ can be replaced by 
\begin{align*}
	\rho_Q(I) \coloneqq \left\{ f \in \rho(I) \mid f \text{ compatible with } Q\right\}.
\end{align*}

The following corollary from~\cite[Theorem 16]{RaabRegensburgerHosseinPoor2021}
ensures under very mild assumptions on the generators of $I$ that, when working with $\rho_Q(I)$ instead of $\rho(I)$, we can still form all compatible polynomials in $I$.

\begin{corollary}
Let $Q$ be a labelled quiver and let $f_1,\dots,f_r,f \in K\<X>$ be compatible with $Q$.
Furthermore, let $I = (f_1,\dots,f_r)$.
Then, 
\begin{align*}
	f \in I \iff f \in \left(\rho_Q(I)\right)_\rho.
\end{align*}
\end{corollary}

Typically, the set $\rho_Q(I)$ is a lot smaller than $\rho(I)$.
To illustrate this point, we finish the proof of Theorem~\ref{thm example 1}.

\begin{remark}
It has been observed that the noncommutative version of Buchberger's algorithm preserves compatibility of polynomials with a quiver~\cite{HeltonStankusWavrik1998,chenavier2020compatible,LevandovskyySchmitz2020}.
This means that, if all polynomials that are given as input to this algorithm are compatible with a quiver, then so are all polynomials in the output.
We note that this behavior carries over to one-sided ideal intersections. 
More precisely, if all generators of two right ideals $I_\rho$ and $J_\rho$ are compatible with a quiver,
then so are all polynomials in a Gr\"obner basis of $I_\rho \cap J_\rho$ when Buchberger's algorithm is used in combination with Theorem~\ref{thm intersect right ideals}.
In case of two-sided ideal intersections, however, this is no longer true. 
A Gr\"obner basis of the intersection $I \cap J$ of two ideals $I$ and $J$ with compatible generators computed with Buchberger's algorithm might contain incompatible polynomials. 
The reason for this is the introduction of the commutator relations $tx_k - x_kt$ during the computation of $I\cap J$.
\end{remark}

\begin{example}
Recall that we had translated the assumptions of Theorem~\ref{thm example 1} into a set $F$ of 10 noncommutative polynomials with integer coefficients in 12 indeterminates.
More precisely, the set $F$ consists of the following elements:
\begin{align*}
	&a a^\dagger a - a, & &a^* (a^\dagger)^* a^* - a^*, & &(a^\dagger)^* a^* - a a^\dagger, \\
	 &a^\dagger a a^\dagger - a^\dagger, & &(a^\dagger)^* a^* (a^\dagger)^* - (a^\dagger)^*, & &a^* (a^\dagger)^* - a^\dagger a,\\
	&c - ay, & &c^* - y^* a^*, & &c^* (a^\dagger)^* - b^* z, & &a^\dagger c - z^* b,
\end{align*}
where the first two lines correspond to the existence of the Moore-Penrose inverse $A^\dagger$ (plus the respective adjoint statements)
and the last line corresponds to the range inclusions $\mathcal{R}(C) \subseteq \mathcal{R}(A)$ and $\mathcal{R}((A^\dagger C)^*) \subseteq \mathcal{R}(B^*)$ (plus the respective adjoint statements).
The quiver encoding the domains and codomains of the operators involved is depicted in Figure~\ref{fig:quiver}.

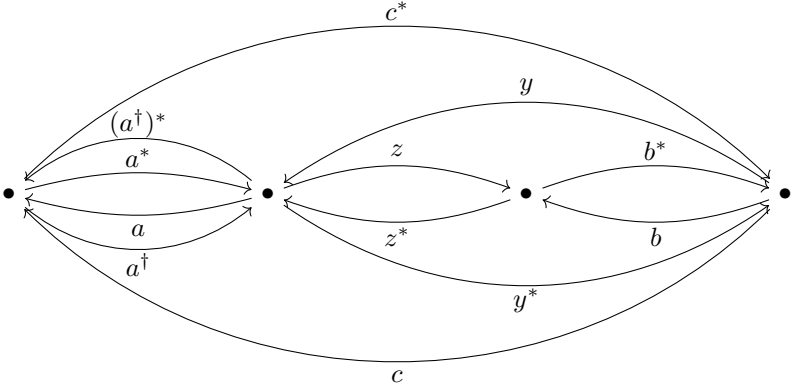
\begin{figure}
 \begin{tikzpicture}[ampersand replacement=\&,baseline=3,xshift=-3em]
    \matrix (m) [matrix of math nodes, column sep=30mm]
       { \bullet \& \bullet \& \bullet \& \bullet \\};
        \path[->] (m-1-2) edge [bend right = -15] node [yshift=-2mm] {$a$} (m-1-1);
    \path[->] (m-1-1) edge [bend right = -15] node [yshift=2mm] {$a^*$} (m-1-2);
     \path[->] (m-1-1) edge [bend left = -40] node [yshift=-2mm] {$a^\dagger$} (m-1-2);
    \path[->] (m-1-2) edge [bend left = -40] node [yshift=2mm] {$(a^\dagger)^*$} (m-1-1);
     \path[->] (m-1-4) edge [bend right = -20] node [yshift=-2mm] {$b$} (m-1-3);
    \path[->] (m-1-3) edge [bend right = -20] node [yshift=2mm] {$b^*$} (m-1-4);
     \path[->] (m-1-4) edge [bend right = -45] node [yshift=-2mm] {$c$} (m-1-1);
    \path[->] (m-1-1) edge [bend right = -45] node [yshift=2mm] {$c^*$} (m-1-4);
     \path[->] (m-1-4) edge [bend left = -35] node [yshift=2mm] {$y$} (m-1-2);
    \path[->] (m-1-2) edge [bend left = -35] node [yshift=-2mm] {$y^*$} (m-1-4);
    \path[->] (m-1-2) edge [bend right = -20] node [yshift=2mm] {$z$} (m-1-3);
    \path[->] (m-1-3) edge [bend right = -20] node [yshift=-2mm] {$z^*$} (m-1-2);
  \end{tikzpicture}
  \caption{Quiver encoding the domains and codomains of the operators of Theorem~\ref{thm example 1}}
  \label{fig:quiver}
\end{figure}

The existence of a solution $X$ of $AXB = C$ corresponds to finding a compatible polynomial $axb - c \in (F)$, where $a,b,c$ are known but $x$ is unknown.
Using one-sided ideal intersection as discussed in this section, we can now find such a polynomial.
To this end, we intersect $I = (F)$ with the right ideal $(a,c)_\rho$.
Working over $K=\mathbb{Q}$ with a degree lexicographic ordering, the reduced Gr\"obner basis of $I$ consists of 17 polynomials.
In this example, $\rho(I)$ as well as $\rho_Q(I)$ seem to be infinite.
However, it turns out that it suffices to consider the finite approximation $\rho^{(5)}_Q(I) \subseteq \rho_Q(I)$ of compatible polynomials of degree at most 5 to find a suitable polynomial in the intersection.
We note that $|\rho^{(5)}_Q(I)| = 263$ while computing all elements in $\rho(I)$ up to degree 5 would yield 3485 elements.
Then, the reduced Gr\"obner basis of $(\rho^{(5)}_Q(I))_\rho \cap (a,c)_\rho$ contains the compatible polynomial $az^*b - c$, which proves Theorem~\ref{thm example 1} and shows that $Z^*$ is a solution of $AXB = C$.
\end{example}

To end this section, we mention another example of applying the methods described above.
In a recent application of the algebraic proof framework, we investigated automated proofs of the triple reverse order law for Moore-Penrose inverses~\cite{Hartwig}.
All parts of these proofs that can be reduced to verifying ideal membership of an explicitly given polynomial could be automatized back then.
The only step that does not lead to such a problem is the proof of certain range inclusions $\mathcal{R}(A) \subseteq \mathcal{R}(B)$.
As explained in the introduction, this leads to the problem of finding a polynomial of the form $a - bx$, where $a,b$ are known but $x$ is an unknown polynomial, in the ideal $(F)$ generated by the assumptions $F$ of the statement.
Back then, a polynomial of this form had to be found by hand.
Using the methods presented in this section, this step can now be supported by computing the intersection of the two-sided ideal $(F)$ with the right ideal generated by $a$ and $b$.
For further information, we refer to the \textsc{Mathematica} notebook accompanying this paper.


\section{Homogeneous part} \label{sec homogeneous part}

The previous section has shown that ideal intersections can provide a useful tool for computing elements of special form in an ideal.
However, ideal intersections require to explicitly know a factor in every summand of the element.
In some situations we do not have this explicit knowledge.
This is for example the case when searching for polynomials in an ideal $I \subseteq K\<X,X^*>$ of the form $p - q^* q \in I$, where $p$ is known but $q$, and therefore also $q^*$, are unknown.
Recall that $q^*$ is the image of $q$ under the involutive antiautomorphism ${}^*$ on  $K\<X,X^*>$ that maps each $x_i \in X$ to $x_i^* \in X^*$ and each $x_i^*$ back to $x_i$. 
As already mentioned in the introduction, such a task corresponds to determining the positivity of the linear operator represented by the polynomial $p$.
If $q$ is a monomial, a polynomial of the desired form can be found by computing the \emph{homogeneous part} of the ideal $I$.

So, in this section, we describe how the homogeneous part of an ideal can be computed. 
Our result is a generalization of ~\cite{Mil16}, which describes how the set of all monomials contained in a commutative ideal can be computed.
We note that a variant of the approach described in ~\cite{Mil16} can also be found in~\cite[Algorithm~4.4.2]{SST00}.

First, we recall that an \emph{$M$-grading} of a ring $R$ by a monoid $(M,+)$ is given by a decomposition
$R = \bigoplus_{\alpha \in M} R_\alpha$
of $R$ into a direct sum of abelian groups $R_\alpha$ such that $R_\alpha R_\beta \subseteq R_{\alpha+\beta}$ for all $\alpha, \beta \in M$.
Given an $M$-grading of $R$, an element $r\in R$ is called \emph{homogeneous} if $r \in R_\alpha$ for some $\alpha \in M$.
If $r \neq 0$, then $\alpha$ is called the \emph{degree} of $r$, denoted by $\deg(r) = \alpha$.
Every nonzero $r \in R$ can be written uniquely as $r = r_{\alpha_1} + \dots + r_{\alpha_d}$ with $\alpha_j \in M$ and $0 \neq r_{\alpha_j} \in R_{\alpha_j}$ for all $j = 1,\dots,d$.
The elements $r_{\alpha_j}$ are referred to as the \emph{homogeneous components} of $r$.
An ideal $I \subseteq R$ is called \emph{homogeneous} if $I$ is generated by homogeneous elements.
We note that there are also other equivalent characterisations of a homogeneous ideal.
The following one will be relevant for the proof of Theorem~\ref{thm homogeneous part}.
An ideal $I \subseteq R$ is homogeneous if and only if, with every nonzero element $r \in I$, also all homogeneous components of $r$ lie in $I$.

A matrix $A\in\RR^{n\times k}$ with rows $a_1,\dots,a_n \in \RR^k$ specifies a grading of the free algebra $K\<X>$ by the monoid $(\RR^k,+)$.
The grading is given by the monoid homomorphism $\deg_A: \<X> \to \RR^k$ defined on the basis $X$ to be $\deg_A(x_i) = a_i$, for $i = 1,\dots,n$.
This map decomposes $K\<X>$ into a direct sum of the abelian subgroups
\begin{align*}
	K\<X>_\alpha = \left\{\sum_{m\in\<X>} c_m m \mid c_m \in K, \deg_A(m) = \alpha \text{ if} c_m \neq 0\right\},
\end{align*}
for $\alpha \in \RR^k$.
Given a grading $\deg_A$ of $K\<X>$ and an ideal $I \subseteq K\<X>$, we can now ask what are the polynomials in $I$ that are homogeneous w.r.t.\ $\deg_A$.
More precisely, we want to consider the ideal generated by all these polynomials.
Analogous to the commutative case (see e.g.~\cite[Tutorial 50]{KR05}), we call this ideal the \emph{homogeneous part} of $I$ w.r.t.\ $\deg_A$.

\begin{definition}
Let $I \subseteq K\<X>$ be an ideal and let $A\in\RR^{n\times k}$.
We denote by
$\hom_A(I) \coloneqq (f \in I \mid f \text{ is homogeneous w.r.t.} \deg_A )$
the \emph{homogeneous part} of $I$ w.r.t.\ $\deg_A$.
\end{definition}

The next lemma follows directly from the definition above.

\begin{lemma}\label{lemma homogeneous ideal}
Let $I \subseteq K\<X>$ be an ideal and let $A\in\RR^{n\times k}$.
Then, $\hom_A(I)$ is the largest homogeneous ideal w.r.t.\ $\deg_A$ contained in $I$.
\end{lemma}

In the following, we describe how a Gr\"obner basis of $\hom_A(I)$ can be computed.
In our approach, we have to restrict ourselves to the case that $A$ has only integer entries, that is, $A \in \ZZ^{n \times k}$.
Note that this is equivalent to having $A \in \QQ^{n \times k}$ as, for any nonzero $c \in \QQ$, the matrix $cA$ induces the same decomposition of $K\<X>$ as $A$.

For the following, some notation is needed.
First, for two sets $Y$ and $Z$, we denote by $[Y,Z] \coloneqq \{yz - zy \mid y\in Y, z\in Z\}$ the set of commutator relations between $Y$ and $Z$.
Then, with two new sets of indeterminates $T = \{t_1,\dots,t_k\}$ and $T^{-1} = \{t_1^{-1},\dots,t_k^{-1}\}$, we let $J \subseteq K\<X,T,T^{-1}>$ be the ideal generated by 
\begin{align*}
	[X \cup T \cup T^{-1},T \cup T^{-1}] \cup \{1 - t_j t_j^{-1} \mid j=1,\dots,k\}.
\end{align*}
Furthermore, we let $\mathcal{A} =  K\<X,T,T^{-1}> / J$.
We denote the equivalence class of an element $f \in K\<X,T,T^{-1}>$ in $\mathcal{A}$ by $\overline f$.
Note that, in $\mathcal{A}$, the elements $\overline t_j$ and $\overline{t_j^{-1}}$ commute with each other and with everything else.
Furthermore, they are also invertible.
Consequently, we can think of the equivalence classes of monomials in the variables $T \cup T^{-1}$ as commutative monomials.
Hence, there is a bijection $\tau$ between exponent vectors and suitable representatives of these equivalence classes.
More precisely, for $\alpha = (\alpha_1,\dots,\alpha_k)\in\ZZ^k$ we define
\begin{align*}
	\tau(\alpha) \coloneqq s_1^{|\alpha_1|} \dots s_k^{|\alpha_k|} \in \<T,T^{-1}>, \quad \text{where } s_j = \begin{cases} t_j & \text{if } \alpha_j \geq 0 \\t_j^{-1} & \text{otherwise}\end{cases}.
\end{align*}

\begin{example}
For $\alpha = (3,-2,0,1) \in \ZZ^4$, we have $\tau(\alpha) = t_1^3 (t_2^{-1})^2 t_4$.
\end{example}

So, the map $\tau$ turns a vector of integers into a monomial in $\<T,T^{-1}>$.
The following lemma captures an important property of this map in $\mathcal{A}$.
It follows from the fact that the equivalence classes $\overline t_j$ commute with each other and with their designated inverses.

\begin{lemma}\label{lemma property tau}
Let $\alpha, \beta \in \ZZ^k$. Then, $\overline{\tau(\alpha+\beta)} = \overline{\tau(\alpha)} \cdot \overline{\tau(\beta)}$.
\end{lemma}

For a fixed matrix $A\in \ZZ^{n \times k}$, we can use the map $\tau$ to define the $K$-algebra homomorphism
\begin{align*}
	\varphi_A: K\<X> \to K\<X,T,T^{-1}>, \quad x_i \mapsto x_i \tau(a_i),
\end{align*}
where $a_i$ denotes the $i$-th row of $A$.
Furthermore, we let $\overline \varphi_A : K\<X> \to \mathcal{A}$ be the composition of $\varphi_A$ with the canonical epimorphism that sends each element in $K\<X,T,T^{-1}>$ to its equivalence class.
Applying $\overline \varphi_A$ to a polynomial $f$ maps each monomial $m$ appearing in $f$ to $\overline{mm_t}$,
where $m_t = \tau(\deg_A(m)) \in \<{T, T^{-1}}>$ is a monomial in the variables $T \cup T^{-1}$ encoding $\deg_A(m)$.
Finally, for an ideal $I \subseteq K\<X>$ with generating set $F \subseteq K\<X>$, we consider the \emph{extension} of $I$ along the homomorphism $\overline \varphi_A$, which we denote by
\begin{align*}
	I^{\overline \varphi_A} \coloneqq \left(\overline\varphi_A(f) \mid f \in F\right) \subseteq \mathcal{A}.
\end{align*}
Note that this definition is independent of the generating set $F$. 
Then, we can state the main result of this section.

\begin{theorem}\label{thm homogeneous part}
Let $I \subseteq K\<X>$ be an ideal and let $A \in \ZZ^{n\times k}$. 
Furthermore, let $I^{\overline \varphi_A}$ be the extension of $I$ along $\overline \varphi_A$.
Then, 
\begin{align*}
	\hom_A(I) = I^{\overline \varphi_A} \cap K\<X> \coloneqq \bigcup_{\overline f\,\in\,I^{\overline \varphi_A}} \overline f \cap K\<X>.
\end{align*}
\end{theorem}

\begin{proof}
We note that it follows from the definition of $I^{\overline \varphi_A}$ that $f \in I$ implies $\overline \varphi_A(f) \in I^{\overline \varphi_A}$.
Now, to prove the inclusion $\hom_A(I) \subseteq I^{\overline \varphi_A} \cap K\<X>$, let $f \in \hom_A(I)$.
W.l.o.g.\ we can assume that $f$ is homogeneous.
Let $\alpha = \deg_A(f)\in \ZZ^k$.
Note that since $f$ is homogeneous, every term in $f$ has degree $\alpha$.
Due to this fact, in $\mathcal{A}$, we have $\overline{\tau(\alpha) f} = \overline \varphi_A(f) \in I^{\overline \varphi_A}$.
To see that this implies that also $\overline f \in I^{\overline \varphi_A}$, we use Lemma~\ref{lemma property tau} and compute
\begin{align*}
	\overline f = \overline{\tau(0)} \cdot \overline f = \overline{\tau(-\alpha)} \cdot \overline{\tau(\alpha)} \cdot  \overline f = \overline{\tau(-\alpha)} \cdot \overline{\tau(\alpha) f} \in I^{\overline \varphi_A}.
\end{align*}

For the other inclusion $I^{\overline \varphi_A} \cap K\<X> \subseteq \hom_A(I)$,
we show that $I^{\overline \varphi_A} \cap K\<X> \subseteq I$ and that $I^{\overline \varphi_A} \cap K\<X>$ is a homogeneous ideal. 
Then, the claim follows from Lemma~\ref{lemma homogeneous ideal}.
For the first part, let $f\in I^{\overline \varphi_A} \cap K\<X>$.
Then, $f$ can be written as 
\[
	f = \sum_{i=1}^d p_i \varphi_A(f_i) q_i + \sum_{i=1}^e u_i g_i v_i,
\]
with $f_i \in I$, $g_i \in [X \cup T \cup T^{-1},T \cup T^{-1}] \cup \{1 - t_j t_j^{-1} \mid j=1,\dots,k\}$ and $p_i ,q_i, u_i, v_i \in K\<X,T,T^{-1}>$.
Note that the left hand side of this identity does not depend on the indeterminates in $T$ and $T^{-1}$.
Hence, by setting $t_j = t_j^{-1} = 1$ for all $j = 1,\dots,k$, we obtain
$f = \sum_{i=1}^d \tilde p_i f_i \tilde q_i$ with $\tilde p_i, \tilde q_i \in K\<X>$.
This shows that $f \in I$.

Finally, to show that $I^{\overline \varphi_A} \cap K\<X>$ is a homogeneous ideal, we extend the $\ZZ^k$\nobreakdash-grading of $K\<X>$ defined by $A$ to a $\ZZ^k$-grading of $K\<X,T,T^{-1}>$ by setting 
\begin{align*}
	\deg_A(x_i) = a_i, \quad \deg_A(t_j) = -e_j,\quad \deg_A(t_j^{-1}) = e_j,
 \end{align*}
 for $i = 1,\dots,n$, $j = 1,\dots,k$, where $e_j$ denotes the $j$-th unit vector of $\ZZ^k$.
 Then, the ideal $J$ is homogeneous w.r.t.\ this grading. 
 Consequently, this induces a well-defined grading on $\mathcal{A}$ w.r.t.\ which the ideal $I^{\overline \varphi_A}$ is homogeneous.
To see that also the intersection $I^{\overline \varphi_A} \cap K\<X>$ is homogeneous, let $f \in I^{\overline \varphi_A} \cap K\<X>$.
Since $I^{\overline \varphi_A}$ is homogeneous, the homogeneous components of $\overline f$ are in $I^{\overline \varphi_A}$.
As the normal forms of these homogeneous components do not depend on $T$ or $T^{-1}$, they are consequently also contained in $I^{\overline \varphi_A}\cap K\<X>$.
This shows that the intersection is also a homogeneous ideal.
\end{proof}

To turn Theorem~\ref{thm homogeneous part} into an effective procedure, we can make use of the fact that there is a bijection between ideals in $K\<X,T,T^{-1}> / J$ and superideals of $J$ in $K\<X,T,T^{-1}>$.
Using this, we can obtain a Gr\"obner basis of $\hom_A(I)$ as described in Procedure~\ref{proc:hom}.

\begin{algorithm}
  \caption{Gr\"obner basis enumeration of $\hom_A(I)$}
  \label{proc:hom}
  \begin{algorithmic}[1]
    \REQUIRE a generating set $F\subseteq K\<X>$ of an ideal $I$, a degree matrix $A \in \ZZ^{n \times k}$
    \ENSURE $G\subseteq I$ a Gr\"obner basis of $\hom_A(I)$
    \STATE $T \leftarrow \{t_1,\dots,t_k\}$, $T^{-1} \leftarrow \{t_1^{-1},\dots,t_k^{-1}\}$
    \STATE $J \leftarrow$ the ideal in $K\<X,T,T^{-1}>$ generated by 
    \[
	[X \cup T \cup T^{-1},T \cup T^{-1}] \cup \{1 - t_j t_j^{-1} \mid j=1,\dots,k\} \vspace{-1.3em}
    \] 
    \STATE $H \leftarrow (\varphi_A(f) \mid f \in F) \subseteq K\<X,T,T^{-1}>$
    \STATE enumerate a Gr\"obner basis $\{g_0,g_1,\dots\}$ of $H + J$ w.r.t.\ an elimination ordering for $T \cup T^{-1}$
    \STATE $G \leftarrow \{g_i \mid  i \in \NN, g_i \in K\<X>\}$
     \RETURN{$G$}
  \end{algorithmic}
\end{algorithm}

\begin{remark}
If $A \in \NN^{n\times k}$, then the ideal $J$ can be replaced by 
\begin{align*}
	J' = \left([X \cup T,T] \cup \{1- t_j t_j^{-1} \mid j = 1,\dots,k\}\right).
\end{align*}
\end{remark}

To end this section, we discuss how computing the homogeneous part of an ideal $I \subseteq K\<X,X^*>$ allows to find elements of the form $p - q^* q \in I$,
where $p \in K\<X,X^*>$ is given and $q$ is unknown.
In case that $q$ is a monomial, we can choose a grading of $K\<X,X^*>$ that makes elements of the form $p - q^* q$ homogeneous.
Taking the matrix $A$ such that
$\deg_A(x_i) = e_i$ and $\deg_A(x_i^*) = -e_i$, where $e_i$ denotes the $i$-th unit vector of $\ZZ^n$ for all $i = 1,\dots,n$, yields $\deg_A(q^*q) = 0$ for all $q\in \<X,X^*>$. 
Now, if there exists an element of the form $p - q^* q \in I$ and if we introduce a new variable $v$ with $\deg_A(v) = 0$ and set $I' = I + (v-p)$,
then the homogeneous polynomial $v - q^*q$ of degree 0 lies in $I'$.
By enumerating generators of $\hom_A(I')$ we can systematically search for this element.
We note that we can increase our chances of finding a suitable element by computing a Gr\"obner basis of $\hom_A(I')$ w.r.t.\ an elimination ordering for $\{v\}$.
Because then, provided that $I$ contains an element of the form $p - q^* q$, this Gr\"obner basis must contain an element with leading monomial $v$.
In the following, we apply this procedure to prove Theorem~\ref{thm example 2}.

\begin{example}
Recall that Theorem~\ref{thm example 2} states that if $\mathcal{R}(B) \subseteq \mathcal{R}(A^*)$ and if the operator equation $AXB = C$ has a positive solution $X$, then the operator $B^* A^\dagger C$ is positive as well.
To prove this statement by a computation with polynomials, we have to translate all properties of the operators first into identities and then into noncommutative polynomials.
For the assumptions, the existence of a positive solution $X$ and the range inclusion translate into the identities $AY^*YB = C$ and $B = A^* Z$, respectively, with new operators $Y,Z$.
Furthermore, the existence of the Moore-Penrose inverse $A^\dagger$ is encoded by the four defining identities of $A^\dagger$.
Translating all these identities and the respective adjoint statements into noncommutative polynomials, gives a set $F \subseteq \QQ\<X,X^*>$ of 10 polynomials with integer coefficients in 12 indeterminates.
Proving Theorem~\ref{thm example 2} boils down to finding an operator $Q$ such that $B^* A^\dagger C = Q^* Q$, or in terms of polynomials, to finding a compatible polynomial
of the form $b^* a^\dagger c - q^* q \in (F)$, where $b^*, a^\dagger, c \in X \cup X^*$ are known but $q \in \QQ\<X,X^*>$, and therefore also $q^*$, are unknown.
We note that the quiver encoding the domains and codomains of the operators involved consists of 3 vertices and 12 edges representing the 12 indeterminates.
We refer to the \textsc{Mathematica} notebook accompanying this paper for a visualization of the quiver and for further information on the computations. 

Following the procedure outlined above, we consider the ideal $I' = (F) + (v - b^* a^\dagger c) \subseteq \QQ\<X,X^*,v>$ and enumerate a Gr\"obner basis of the homogeneous part $\hom_A(I')$ w.r.t.\ the degree matrix $A \in \ZZ^{13 \times 6}$ such that $\deg_A(v) = 0$,
$\deg_A(x_i) = e_i$, $\deg_A(x_i^*) = - e_i$ for all $x_i \in X$, $x_i^* \in X^*$.
We do this computation w.r.t.\ an elimination ordering for $\{v\}$.
To speed up the computation, we first compute the reduced Gr\"obner basis $G$ of $I'$ and then use this generating set as input to enumerate a Gr\"obner basis of $\hom_A(I')$. 
While $|G| = 25$, the Gr\"obner basis of $\hom_A(I')$ seems to be infinite.
However, after enumerating 273 generators of $\hom_A(I')$, one can see that $v - b^*y^*yb = v - (yb)^* yb  \in \hom_A(I')$, which shows that $b^* a^\dagger c - (yb)^* yb \in (F)$.
After verifying that this polynomial is compatible with our quiver, this reveals that $B^* A^\dagger C = (YB)^* YB$ is indeed positive, and consequently, proves Theorem~\ref{thm example 2}.
\end{example}


\section{Monomial part}\label{sec monomial part}

Other elements in an ideal that are often of special interest are monomials. 
The goal of this section is to effectively describe the \emph{monomial part} of a (right) ideal $I \subseteq K\<X>$, that is, the (right) ideal generated by all monomials contained in $I$.
More precisely, we show how a simple adaptation of the technique described in the previous section allows to compute all monomials in a given right ideal.
We note that the following definition is analogous to the commutative case (cf.\ the discussion before Algorithm~4.4.2 in~\cite{SST00}). 

\begin{definition}
Let $I \subseteq K\<X>$ be an ideal.
Then, we denote by $\mon(I) \coloneqq (m \in I \mid m \in \<X> )$ the \emph{monomial part} of $I$.
Analogously, for a right ideal $I_\rho \subseteq K\<X>$, we denote by $\mon(I_\rho) \coloneqq (m \in I \mid m \in \<X> )_\rho$ the \emph{monomial part} of $I_\rho$.
\end{definition}

This is clearly a monomial (right) ideal as defined below.

\begin{definition}
Let $I \subseteq K\<X>$ be a (right) ideal.
Then, $I$ is called a \emph{monomial (right) ideal} if $I$ has a (right) generating set consisting only of monomials. 
\end{definition}

Monomial ideals can also be characterized in a different way.
To this end, we recall that the \emph{support} of a polynomial $f \in K\<X>$, denoted by $\supp(f)$, is the set consisting of all monomials that appear in $f$.
Then, a (right) ideal $I \subseteq K\<X>$ is a monomial (right) ideal if and only if, with every element $f \in I$, also $\supp(f) \subseteq I$.
The next lemma follows directly from the definition above.

\begin{lemma}\label{lemma monomial ideal}
Let $I \subseteq K\<X>$ be a (right) ideal.
Then, $\mon(I)$ is the largest monomial (right) ideal contained in $I$.
\end{lemma}

In the following, we describe how an adaptation of the technique from the previous section allows to compute a right Gr\"obner basis of $\mon(I_\rho)$ for a given right ideal $I_\rho \subseteq K\<X>$.
As for the homogeneous part of an ideal, we need some notation.
First, we define the ideal 
\begin{align}\label{quotient ideal}
	J = \left([X,T] \cup \{1 - t_j t_j^{-1} \mid j = 1,\dots, n\}\right),
\end{align}
where $T = \{t_1,\dots,t_n\}$ and $T^{-1} = \{t_1^{-1},\dots,t_n^{-1}\}$ are new indeterminates.
Then, we consider the quotient algebra $\mathcal{B} =  K\<X,T,T^{-1}>/ J$.
We again denote the equivalence classes of elements $f \in K\<X,T,T^{-1}>$ in $\mathcal{B}$ by $\overline f$.
In $\mathcal{B}$, the $\overline t_j$ commute with all $\overline x_i$ but not with each other.
The fact that the $\overline t_j$ do not commute with each other is the main and crucial difference compared to the previous section.
Furthermore, each $\overline t_j$ can also be cancelled from the right.

Additionally, we define the $K$-algebra homomorphism 
\begin{align}\label{def phi monomial}
	\varphi: K\<X> \to K\<X,T,T^{-1}>, \quad x_i \mapsto x_i t_i.
\end{align}
The map $\varphi$ can be considered as a special case of the map $\varphi_A$ defined in the previous section by setting $A = I_n$.
As before, $\overline \varphi$ denotes the composition of $\varphi$ with the canonical epimorphism.
Applying $\overline \varphi$ to a polynomial $f$ maps each monomial $m\in \supp(f)$ to the equivalence class $\overline{mt_m}$, where $t_m$ is a copy of $m$ but in the variables $t_1,\dots,t_n$.

Next, for a right ideal $I_\rho \subseteq K\<X>$ with right generating set $F$, we consider the extension of $I_\rho$ along the homomorphism $\overline \varphi$, which we denote by
\begin{align}\label{def tI}
	I_\rho^{\overline \varphi} \coloneqq (\overline \varphi(f) \mid f \in F)_\rho \subseteq \mathcal{B}.
\end{align}
Note that, as in the previous section, this definition is independent of the generating set $F$.
However, in contrast to before,  $I_\rho^{\overline \varphi}$ is now a right ideal.
The main result of this section is the following theorem.

\begin{theorem}\label{thm monomials in ideal}
Let $I_\rho \subseteq K\<X>$ be a right ideal. 
Furthermore, let $I_\rho^{\overline \varphi}$ be the extension of $I_\rho$ along $\overline \varphi$. 
Then,
\begin{align*}	
	\mon(I_\rho) = I_\rho^{\overline \varphi} \cap K\<X> \coloneqq \bigcup_{\overline f \,\in\, I_\rho^{\overline \varphi}} \overline f \cap K\<X>.
\end{align*}
\end{theorem}

Before we proceed to prove this theorem, we describe in Procedure~\ref{proc:mon} how a right Gr\"obner basis of $\mon(I_\rho)$ can be obtained. 

\begin{remark}
Procedure~\ref{proc:mon} requires the computation of a right Gr\"obner basis of the right ideal $I_\rho^{\overline \varphi}$ in the quotient algebra $\mathcal{B} = K\<X,T,T^{-1}> / J$.
As described in~\cite{Hey01}, this can be done effectively by introducing a new tag variable $y$ and computing the Gr\"obner basis of a two-sided ideal containing $J$ in
the free algebra $K\<X,T,T^{-1},y>$.
This is what is done in line~\ref{line:mon1} -- \ref{line:mon2} of Procedure~\ref{proc:mon}.
We note that there are also other ways to compute a Gr\"obner basis of a one\nobreakdash-sided ideal in a quotient algebra (see e.g.~\cite[Section~6.1.2]{Xiu12}),
however the approach described above has the advantage that it can be realized with a standard two-sided Gr\"obner basis implementation without any additional adaptations. 
\end{remark}

\begin{algorithm}
  \caption{Gr\"obner basis enumeration of $\mon(I_\rho)$}
  \label{proc:mon}
  \begin{algorithmic}[1]
    \REQUIRE a right generating set $F\subseteq K\<X>$ of a right ideal $I_\rho$
    \ENSURE $G\subseteq I_\rho$ a right Gr\"obner basis of $\mon(I_\rho)$
    \STATE $T \leftarrow \{t_1,\dots,t_n\}$, $T^{-1} \leftarrow \{t_1^{-1},\dots,t_n^{-1}\}$
    \STATE $J \leftarrow$ the ideal in $K\<X,T,T^{-1},y>$ generated by 
    \[
	 [X,T] \cup \{1 - t_j t_j^{-1} \mid j = 1,\dots, n\} \vspace{-1.3em}
    \] 
    \STATE $H \leftarrow (y \cdot \varphi(f) \mid f \in F) \subseteq K\<X,T,T^{-1},y>$ \label{line:mon1}
    \STATE enumerate a Gr\"obner basis $\{g_0,g_1,\dots\}$ of $H + J$ w.r.t.\ an elimination ordering for $T \cup T^{-1}$
    \STATE $G' \leftarrow \{g \mid y g = g_i \text{ for some } i \in \NN \}$ \label{line:mon2} 
    \STATE $G \leftarrow G' \cap K\<X>$
     \RETURN{$G$}
  \end{algorithmic}
\end{algorithm}

The proof of  the commutative analog of Theorem~\ref{thm monomials in ideal} in \cite{Mil16} relies on the fact that a certain ideal is homogeneous w.r.t.~a certain matrix grading and that an ideal can only be homogeneous w.r.t.~this grading if it is a monomial ideal.
Unfortunately, this argument does not immediately carry over to the case of noncommutative polynomials because no matter which matrix grading we choose, monomials that are permutations of each other will always have the same degree. 
For example, the right ideal $(x_1 x_2 -x_2 x_1)_\rho$ is homogeneous w.r.t.\ any matrix grading but it is clearly not a monomial ideal. 
In order to prove Theorem~\ref{thm monomials in ideal}, we introduce the notion of a \emph{separating pseudograding} for a subset $S \subseteq R$ of a ring $R$.

\begin{definition}\label{def s-grading}
Let $R$ be a ring and let $(N,\cdot)$ be a monoid.
Furthermore, let $S \subseteq R$. 
We call a decomposition
\begin{align*}
	R = \bigoplus_{\alpha \in N} S_\alpha,
\end{align*}
of $R$ into a direct sum of abelian groups a \emph{separating (right) $N$-pseudograding} for $S$ if the following conditions hold:
\begin{enumerate}
	\item $S_1 S_\alpha \subseteq S_\alpha$ for all $\alpha \in N$;
	\item $|S \cap S_\alpha | = 1$ for all $\alpha \in N$;
\end{enumerate}
\end{definition}

Intuitively speaking, a separating $N$-pseudograding for $S$ allows us to separate the elements in $S$ as they all have to lie in different subgroups $S_\alpha$.
As in the case of usual gradings, we omit the information about the monoid $N$ in a separating $N$-pseudograding if $N$ is clear from the context.

The following example shows that the multidegree yields a separating pseudograding for commutative monomials.

\begin{example}
Let $R = K[X]$ be the usual polynomial ring in $X = \{x_1,\dots,x_n\}$.
Consider the monoid $(\NN^n, +)$ and let $S = [X]$.
For $\alpha = (\alpha_1,\dots,\alpha_n) \in \NN^n$, we denote $x^\alpha \coloneqq x_1^{\alpha_1}\dots x_n^{\alpha_n} \in [X]$.
The grading of $R$ induced by the multidegree $\deg(x^\alpha) = \alpha$ is a separating $\NN^n$-pseudograding for $S$.
This follows from the fact that for all $\alpha = (\alpha_1,\dots,\alpha_n) \in \NN^n$, the set $S_\alpha$ is given by $S_\alpha = \{c x^\alpha \mid c \in K \setminus \{0\}\}$. 
\end{example}

For the following definition and proposition, we fix a ring $R$ and a monoid $(N,\cdot)$.
Given a separating ($N$-)pseudograding for a set $S \subseteq R$, we can define \emph{separable homogeneous} elements, or short \emph{s-homogeneous} elements, and, based on this, \emph{separable homogeneous} ideals, or short \emph{s-homogeneous} ideals.

\begin{definition}
Let $S \subseteq R$ and let $R = \bigoplus_{\alpha \in N} S_\alpha$ be a separating pseudograding for $S$.
An element $r\in R$ is called \emph{separable homogeneous}, or short \emph{s-homogeneous}, if $r \in S_\alpha$ for some $\alpha \in N$.
If $r \neq 0$, then we call $\alpha$ the \emph{s-degree} of $r$, denoted by $\sdeg(r) = \alpha$.
An ideal $I \subseteq R$ is called \emph{separable homogeneous}, or short \emph{s-homogeneous}, if $I$ is generated by s-homogeneous elements.
\end{definition}

The following proposition is the crucial observation about separating pseudogradings.
It is a weaker version of the property that, with an element $r$, a homogeneous ideal also contains all homogeneous components of $r$.

\begin{proposition}\label{prop s-grading}
Let $S \subseteq R$ and $R = \bigoplus_{\alpha \in N} S_\alpha$ be a separating pseudo\-grading for $S$.
Furthermore, let $I_\rho \subseteq R$ be an s-homogeneous right ideal generated by s\nobreakdash-homogeneous elements of s-degree 1.
If $r = c_1 s_1 + \dots + c_d s_d \in I_\rho$ with $s_1,\dots,s_d \in S$ pairwise different and $c_1,\dots,c_d \in S_1$, then $c_1s_1,\dots, c_ds_d \in I_\rho$.
\end{proposition}

\begin{proof}
W.l.o.g.\ we can assume that $c_i s_i \neq 0$ for all $i = 1,\dots,d$.
We note that it suffices to show that $c_1s_1 \in I_\rho$, for then also $r - c_1s_1 \in I_\rho$ and the statement follows by induction.
We denote by $F \subseteq S_1$ a right generating set of $I_\rho$ consisting of s-homogeneous elements of s-degree 1.
Then, with $b_f \in R$ and $b_{f,\alpha} \in S_\alpha$ such that $b_f = \sum_{\alpha \in N} b_{f,\alpha}$, we can write $r$ as 
\begin{align*}
	r = \sum_{f \in F} f b_f = \sum_{f \in F} f \sum_{\alpha \in N} b_{f,\alpha} = \sum_{f \in F} \sum_{\alpha \in N} f b_{f,\alpha}.
\end{align*}
As all $f \in S_1$, it follows from the definition of a separating pseudograding that $f b_{f,\alpha} \in S_1 S_\alpha \subseteq S_\alpha$.
Furthermore, by the definition of a separating pseudograding for $S$ there exist pairwise different $\alpha_1,\dots,\alpha_d \in N$ such that $S \cap S_{\alpha_i} = \{s_i\}$ for all $i = 1,\dots,d$.
Then, also $c_i s_i \in S_1 S_{\alpha_i} \subseteq S_{\alpha_i}$ for all $i = 1,\dots,d$. 
Since the subgroups $S_{\alpha_i}$ have trivial intersection with each other, the nonzero elements $c_2s_2, \dots, c_d s_d$ cannot lie in $S_{\alpha_1}$. 
Therefore, by comparing s-degrees, one can see that $c_1s_1 = \sum_{f \in F} f b_{f,\alpha_1} \in I_\rho$.
\end{proof}

Coming back to noncommutative polynomials, our next goal is to define a separating pseudograding for $\<X>$ in $K\<X,T,T^{-1}>$.
To this end, we denote by $\mathcal{F}_n$ the free group of rank $n$ generated by $g_1,\dots,g_n$.
This allows us to consider the two monoid homomorphisms
\begin{align*}
	\lambda&: \<X,T,T^{-1}> \to \mathcal{F}_n, &&x_i \mapsto g_i, &&t_i \mapsto 1, &&t_i^{-1} \mapsto 1, \\
	\mu&: \<X,T,T^{-1}> \to \mathcal{F}_n,  &&x_i \mapsto 1, &&t_i \mapsto g_i, &&t_i^{-1} \mapsto g_i^{-1}. 
\end{align*}
Using $\lambda$ and $\mu$, we define the map $\sdeg: \<X,T,T^{-1}> \to \mathcal{F}_n$ that sends each $m \in \<X,T,T^{-1}>$ to $\sdeg(m) \coloneqq \mu(m)^{-1} \lambda(m) \in \mathcal{F}_n$.

\begin{example}
Let $X = \{x_1,x_2\}$, $T = \{t_1,t_2\}$ and $T^{-1} = \{t_1^{-1},t_2^{-1}\}$.
Then, for $m = x_1 t_1 x_2 t_2 \in \<X,T,T^{-1}>$, we have $\sdeg(m) = (g_1 g_2 )^{-1} g_1g_2 = 1$.
Analogously, taking $m' = t_1 t_2^{-1} x_1 t_1^{-1} x_2 \in \<X,T,T^{-1}>$, yields $\sdeg(m') = (g_1 g_2^{-1} g_1^{-1})^{-1} g_1 g_2 = g_1 g_2 g_2$.
\end{example}

The following lemma captures some important properties of the map $\sdeg$.

\begin{lemma}\label{lemma properties sdeg}
Let $m, m' \in \<X,T,T^{-1}>$, $x_i \in X$, $t_j \in T$ and $t_j^{-1} \in T^{-1}$. 
Then, the following hold.
\begin{enumerate}
	\item $\sdeg(mx_it_jm') = \sdeg(mt_jx_im')$;
	\item $\sdeg(mt_jt_j^{-1}m') = \sdeg(mm')$;
	\item $\sdeg(m) = 1 \implies \sdeg(mm') = \sdeg(m')$;
\end{enumerate}
\end{lemma}

\begin{proof}
The properties 1 and 2 follow immediately from the definition.
To prove the third property, we assume $\sdeg(m) = 1$.
This yields
\begin{align*}
	\sdeg(mm') &=  \mu(mm')^{-1} \lambda(mm') = \mu(m')^{-1} \mu(m)^{-1} \lambda(m) \lambda(m') \\
	&= \mu(m')^{-1} \sdeg(m) \lambda(m') = \mu(m')^{-1} \lambda(m') = \sdeg(m').
\end{align*}
\end{proof}

The map $\sdeg$ allows us to decompose $K\<X,T,T^{-1}>$ into a direct sum of the abelian subgroups
\begin{align*}
	K\<X,T,T^{-1}>_\alpha = \left\{\sum_{m\in\<X,T,T^{-1}>} c_m m \mid c_m \in K, \sdeg(m) = \alpha \text{ if} c_m \neq 0\right\},
\end{align*}
for $\alpha \in \mathcal{F}_n$. 
In particular, the following proposition tells us that this decomposition forms a separating $\mathcal{F}_n$-pseudograding for $\<X>$, and hence, justifies the name $\sdeg$. 

\begin{proposition}\label{prop sdeg}
The abelian subgroups $K\<X,T,T^{-1}>_\alpha$ with $\alpha \in \mathcal{F}_n$ as defined above form a separating $\mathcal{F}_n$-pseudograding for $\<X>$.
\end{proposition}

\begin{proof}
It is clear that the subgroups $K\<X,T,T^{-1}>_\alpha$ decompose $K\<X,T,T^{-1}>$ into a direct sum. 
Furthermore, the first condition of Definition~\ref{def s-grading} follows immediately from Lemma~\ref{lemma properties sdeg}.
Finally, we note that $\sdeg$ restricted to $\<X>$ is an isomorphism. 
Hence, $\sdeg(m) \neq \sdeg(m')$ for all $m \neq m' \in \<X>$.
This implies the second property of a separating $\mathcal{F}_n$-pseudograding.
\end{proof}

It follows from the first two properties in Lemma~\ref{lemma properties sdeg} that the map $\sdeg$ is invariant modulo the ideal $J$ from \eqref{quotient ideal}.
Hence, we can set  $\sdeg(\overline m) = \sdeg(m)$ for all $m \in \<X,T,T^{-1}>$ and thereby extend $\sdeg$ to the quotient algebra $\mathcal{B}$.
Consequently, the separating $\mathcal{F}_n$-pseudograding for $\<X>$ in $K\<X,T,T^{-1}>$ can be extended in a straightforward way to give a separating
$\mathcal{F}_n$\nobreakdash-pseudograding for $\< \overline X> = \{\overline x \mid x \in \<X>\}$ in $\mathcal{B}$.

Finally, we have all tools available to prove Theorem~\ref{thm monomials in ideal}.
We split this proof into two lemmas.

\begin{lemma}\label{proof monomial in ideal 1}
Let $I_\rho \subseteq K\<X>$ be a right ideal. 
Furthermore, let $I_\rho^{\overline \varphi}$ be the extension of $I_\rho$ along $\overline \varphi$. 
Then,
\begin{align*}	
	\mon(I_\rho) \;\subseteq\; I_\rho^{\overline \varphi} \cap K\<X> \;\subseteq\; I_\rho.
\end{align*}
\end{lemma}

\begin{proof}
We note that it follows from the definition of $I_\rho^{\overline \varphi}$ that $f \in I_\rho$ implies $\overline \varphi(f) \in I_\rho^{\overline \varphi}$.
Now, to prove the first inclusion, $\mon(I_\rho) \subseteq I_\rho^{\overline \varphi} \cap K\<X>$, let $m =  x_{i_1}\dots x_{i_k} \in \mon(I_\rho)$. 
Then, we get $\overline{mt_{i_1}\dots t_{i_k}} = \overline \varphi(m) \in I_\rho^{\overline \varphi}$, which implies
\begin{align*}
	\overline m = \overline {m} \cdot \overline{t_{i_1}\dots t_{i_k} t_{i_k}^{-1}\dots t_{i_1}^{-1}} = \overline{mt_{i_1}\dots t_{i_k}} \cdot \overline{t_{i_k}^{-1}\dots t_{i_1}^{-1}} \in I_\rho^{\overline \varphi},
\end{align*}
and consequently, $m \in  I_\rho^{\overline \varphi} \cap K\<X>$.

For the second inclusion, $I_\rho^{\overline \varphi} \cap K\<X> \subseteq I_\rho$, let $f \in I_\rho^{\overline \varphi} \cap K\<X>$.
Then, $f$ can be written as 
\[
	f = \sum_{i=1}^d \varphi(f_i) q_i + \sum_{i=1}^e g_i v_i 
\]
with $f_i \in I_\rho$, $g_i \in [X,T] \cup \{1 - t_jt_j^{-1} \mid j=1,\dots,n\}$ and $q_i, v_i \in K\<X,T,T^{-1}>$. 
Note that the left hand side of this identity does not depend on the indeterminates in $T$ and $T^{-1}$.
Hence, by setting $t_j = t_j^{-1} = 1$ for all $j = 1,\dots,n$, we obtain $f = \sum_{i=1}^d f_i \tilde q_i \in I_\rho$ with $\tilde q_i \in K\<X>$.
\end{proof}

\begin{lemma}\label{proof monomial in ideal 2}
Let $I_\rho \subseteq K\<X>$ be a right ideal. 
Furthermore, let $I_\rho^{\overline \varphi}$ be the extension of $I_\rho$ along $\overline \varphi$.
Then, $I_\rho^{\overline \varphi} \cap K\<X>$ is a monomial right ideal.
\end{lemma}

\begin{proof}
Let $f = c_1m_1 + \dots + c_d m_d \in I_\rho^{\overline \varphi} \cap K\<X>$ with nonzero $c_1,\dots,c_d \in K$ and pairwise different $m_1,\dots,m_d \in \<X>$.
To show that $I_\rho^{\overline \varphi} \cap K\<X>$  is a monomial right ideal, we have to show that $m_1,\dots,m_d \in I_\rho^{\overline \varphi} \cap K\<X>$, which is equivalent to $\overline m_1,\dots,\overline m_d \in I_\rho^{\overline \varphi}$.
To show this, we fix the $\mathcal{F}_n$\nobreakdash-pseudograding for $\< \overline X>$ in $\mathcal{B}$ induced by the map $\sdeg$ and note that the image $\overline \varphi(g)$ of any nonzero $g \in K\<X>$ under the map $\overline \varphi$ defined in (\ref{def phi monomial}) is an s\nobreakdash-homogeneous polynomial of s-degree 1. 
In particular, this means that $I_\rho^{\overline \varphi}$ is an s\nobreakdash-homogeneous right ideal in $\mathcal{B}$ generated by s\nobreakdash-homogeneous elements of s\nobreakdash-degree 1.
Since $\sdeg(\overline c_i) = 1$ and $\overline m_i \in \<\overline X>$ for all $i = 1,\dots,d$, Proposition~\ref{prop s-grading} is applicable and yields that $\overline{c_1m_1},\dots,\overline{c_dm_d} \in I_\rho^{\overline \varphi}$.
But then clearly also $\overline m_1, \dots, \overline m_d \in I_\rho^{\overline \varphi}$.
\end{proof}

The proof of Theorem~\ref{thm monomials in ideal} now follows from Lemma \ref{lemma monomial ideal}, Lemma \ref{proof monomial in ideal 1} and Lemma \ref{proof monomial in ideal 2}.

\begin{remark} We make some remarks on Theorem~\ref{thm monomials in ideal}.
\begin{enumerate}
	\item Theorem~\ref{thm monomials in ideal} can also be adapted to work with left ideals.
	In this case, the ideal $J$, w.r.t.\ which the quotient is taken, has to be changed to
	\begin{align*}
		J' =  \left([X,T] \cup \{1 - t_j^{-1} t_j \mid j = 1,\dots,n\}\right).
	\end{align*}
	Then, in the quotient $K\<X,T,T^{-1}>/ J'$ each $\overline t_j$ can be cancelled from the left instead of from the right. 
	To then prove the theorem for left ideals, Definition~\ref{def s-grading} and the map $\sdeg$ have to be adapted accordingly as well.
	\item Unfortunately, this approach does not generalize to two-sided ideals, as witnessed by the following example.
\end{enumerate}
\end{remark}

\begin{example}
In this example, we show that the approach described in this section does not generalize to two-sided ideals in a straightforward way.
To this end, we consider the ideal $I = (x_1-x_2) \subseteq K\<x_1,x_2>$ over an arbitrary field $K$.
It is not hard to see that $\mon(I) = \{0\}$.
However, if we consider the two-sided ideal $I^{\overline \varphi} = \left(\overline \varphi(x_1 - x_2)\right) = (\overline{ x_1 t_1 - x_2 t_2}) \subseteq \mathcal{B}$,
then $x_1x_2 - x_2x_1\in I^{\overline \varphi} \cap K\<x_1,x_2>$ as the following computation shows:
\begin{align*}
	&\left(\overline{x_1 t_1 - x_2 t_2}\right) \overline{x_2t_1^{-1}} - \overline {x_2} \left(\overline {x_1 t_1 - x_2 t_2} \right)\overline{t_1^{-1}} \\
	=\; &\left(\overline{x_1x_2} - \overline{x_2x_2t_2t_1^{-1}}\right) - \left(\overline{x_2x_1} - \overline{x_2x_2t_2t_1^{-1}}\right) \\
	=\; &\overline{x_1x_2 - x_2x_1} \in I^{\overline \varphi}.
\end{align*}
\end{example}

To compute a generating set of $\mon(I)$ in case of a two-sided ideal $I \subseteq K\<X>$, we were so far only able to prove the following rather restrictive result.
In what follows, we fix a monomial ordering $\preceq$ on $\<X>$.
Recall that the \emph{tail} of $f \in K\<X>$ is $\tail(f) = f - \lc(f)\lm(f)$, where $\lc(f)$ denotes the leading coefficient of $f$ (w.r.t.\ $\preceq$).

\begin{proposition}\label{prop monomials two sided}
Let $X,Y$ be two disjoint sets of indeterminates and let $I \subseteq K\<X,Y>$ be an ideal.
Furthermore, let $G \subseteq K\<X,Y>$ be a Gr\"obner basis of $I$ such that $\lm(g) \in \<X>$ and $\tail(g)\in K\<Y>$ for all $g \in G$.
Then $G \cap \<X>$ is a Gr\"obner basis of $\mon(I)$.
\end{proposition}

To prove this proposition, we recall the concept of polynomial (top) reduction in $K\<X>$.
For $f,f',g \in K\<X>$ with $f,g \neq 0$, we write $f \rightarrow_g f'$ if there exist $a,b \in \<X>$ such that $\lm(f) = a\lm(g)b$ and
$f' = f - \frac{\lc(f)}{\lc(g)} a g b$. 
It is well\nobreakdash-known that a generating set $G \subseteq K\<X>$ of an ideal $I \subseteq K\<X>$ is a Gr\"obner basis of $I$ if for every $f \in I$,
there exist $g_1,\dots,g_k \in G$ and $f_0,\dots,f_k \in K\<X>$ such that $f_0 = f$, $f_k = 0$ and $f_{i-1} \rightarrow_{g_i} f_i$ for all $i = 1,\dots,k$.
Additionally, we also need the following result about Gr\"obner bases of monomial ideals, which follows immediately from our definition of a Gr\"obner basis. 

\begin{lemma}\label{lemma monomial ideal 2}
Let $M\subseteq \<X>$ be a set of monomials generating an ideal $(M) \subseteq K\<X>$. Then, $M$ is a Gr\"obner basis of $(M)$. 
\end{lemma}

\begin{proof}[Proof of Proposition~\ref{prop monomials two sided}]
Denote $M = G \cap \<X>$.
Due to Lemma~\ref{lemma monomial ideal 2}, it suffices to show that $(M) = \mon(I)$.
The inclusion $(M) \subseteq \mon(I)$ is trivial.
For the other inclusion $\mon(I) \subseteq (M)$, let $m\in\mon(I)$ be a monomial.
Since $G$ is a Gr\"obner basis of $I$, there exist $g_1,\dots,g_k \in G$ such that
\begin{align*}
	m \rightarrow_{g_1} m_1 \rightarrow_{g_2} \dots \rightarrow_{g_k} m_k = 0,
\end{align*}
with $m_1,\dots,m_k \in K\<X,Y>$.
In particular, we have that $m_1 = -a\tail(g_1)b$ for some $a,b\in\<X>$.
Since $\lm(g_2) \in \<X>$ but $\tail(g_1) \in K\<Y>$, the second reduction can only act on $a$ or $b$ but not on $\tail(g_1)$.
W.l.o.g. assume that the reduction with $g_2$ acts on $a$.
Then, $m_2 = a'\tail(g_2)a''\tail(g_1)b$.
Analogously, since $\lm(g_3)\in\<X>$, the reduction with $g_3$ cannot act on $\tail(g_1)$ or $\tail(g_2)$ and as this also holds for all other reductions, $m_k$ must be of the form
\begin{align*}
m_k = (-1)^k a_1\tail(g_{i_1})a_2\tail(g_{i_2})\dots a_k\tail(g_{i_k})a_{k+1}
\end{align*}
for some $a_1,\dots,a_{k+1}\in\<X>$ and $\{i_1,\dots,i_k\}= \{1,\dots,k\}$.
Note that all parts in any $m_j$ that lie in $K\<X>$ are subwords of $m$.
Furthermore, since the reductions with all $g_i$ can only act on these parts, we get that $\lm(g_i)$ divides $m$ for all $i = 1,\dots,k$.
Hence, in particular, $\lm(g_k)$ divides $m$.
To finish the proof, we show that $g_k$ is in fact a monomial and therefore contained in $M$.
To see this, we note that $m_k$ can only be zero if one of the factors in this product is zero and since none of the $a_j$ can be zero,
we must have $\tail(g_i) = 0$ for some $1\leq i \leq k$.
In particular, we know that $m_{k-1} \neq 0$, and therefore, $\tail(g_k) = 0$.
This means that $g_k$ is a monomial, and so, $g_k \in M$, which implies that $m \in (M)$.
\end{proof}

It remains open how to compute a generating set of $\mon(I)$ for an arbitrary (finitely generated) two-sided ideal $I \subseteq K\<X>$.


\subsection*{Acknowledgment}

We are greatly indebted to Dragana S.~Cvetkovi\'c-Ili\'c for discussions on operator equations and for pointing us to the problems considered in~\cite{Arias},
which was one of the motivations of this work.
Additionally, we thank the anonymous referees of the extended abstract presented at CASC 2021 for pointing us to related references and topics.
We are also grateful to the anonymous referees of this paper for their valuable suggestions which helped to improve the presentation of this work a lot.



\begin{thebibliography}{99}
 
\bibitem{Arias}
Arias, M.~L.~and Gonzalez, M.~C.,
\emph{Positive solutions to operator equations $AXB = C$},
Linear Algebra Appl.\ \textbf{433}, pp.~1194--1202, 2010.


\bibitem{BHL17}
Bell, J., Heinle, A., and Levandovskyy, V.,
\emph{On noncommutative finite factorization domains},
Trans.\ Amer.\ Math.\ Soc.\ \textbf{369}, pp.~2675--2695, 2017.


\bibitem{BB98}
Borges, M.~A.~and Borges, M.,
\emph{Groebner Bases Property on Elimination Ideal in the Noncommutative Case},
in \emph{Gröbner Bases and Applications}, pp.~323--337, Cambridge University Press, 1998.

\bibitem{chenavier2020compatible}
Chenavier, C., Hofstadler, C., Raab, C.~G., and Regensburger, G.,
\emph{Compatible Rewriting of Noncommutative Polynomials for Proving Operator Identities},
in \emph{Proceedings of the 45th International Symposium on Symbolic and Algebraic Computation}, pp.~83--90, 2020

\bibitem{Hartwig}
Cvetkovi\'c-Ili\'c, D.~S., Hofstadler, C., Hossein~Poor, J., Milo\v{s}evi\'c, J., Raab, C.~G., and Regensburger, G.,
\emph{Algebraic proof methods for identities of matrices and operators: improvements of Hartwig's triple reverse order law},
Appl.\ Math.\ Comput.\ \textbf{409}, Article 126357, 10 pages, 2021.
 
 \bibitem{Gre00}
  Green, E.~L.,
  \emph{Multiplicative bases, Gr{\"o}bner bases, and right Gr{\"o}bner bases},
  J.\ Symbolic Comput.\ \textbf{29}, pp.~601--623, 2000.

\bibitem{HeltonStankusWavrik1998}
Helton, J.~W., Stankus, M., and Wavrik, J.~J., 
\emph{Computer simplification of formulas in linear systems theory},
 IEEE Trans.\ Automat.\ Control \textbf{43}, pp.~302--314, 1998.

\bibitem{HeltonWavrik1994}
 Helton, J.~W.~and Wavrik, J.~J.,
 \emph{Rules for computer simplification of the formulas in operator model theory and linear systems},
 in \emph{Nonselfadjoint operators and related topics}, pp.~325--354, Birkh\"auser, Basel, 1994.
 
\bibitem{Hey01}
  Heyworth, A.,
  \emph{One-sided noncommutative Groebner bases with applications to computing Green's relations},
  J.\ Algebra \textbf{242}, pp.~401--416, 2001.

\bibitem{Hof20}
  Hofstadler, C.,
  \emph{Certifying operator identities and ideal membership of noncommutative polynomials},
  Master's thesis, Johannes Kepler University Linz, 2020.
  
 \bibitem{KR05}
 Kreuzer, M.~and Robbiano, L.,
 \emph{Computational Commutative Algebra 2},
 Springer, Berlin, 2005.
  
  
 \bibitem{letterplace}
Levandovskyy, V., Sch\"onemann, H., and Abou Zeid, K., 
\emph{Letterplace - a subsystem of Singular for computations with free algebras via Letterplace embedding},
in \emph{Proceedings of the 45th International Symposium on Symbolic and Algebraic Computation}, pp.~305--311, 2020.

\bibitem{Mil16}
Miller, E.,
 \emph{Finding all monomials in a polynomial ideal},
   arXiv:1605.08791, 2016.
   
 \bibitem{Mor87}
 Mora, T.,
 \emph{Gr\"obner bases and the word problem},
 Preprint, University of Genova, 1987.

\bibitem{Mora1994}
Mora,  T.,
 \emph{An introduction to commutative and noncommutative {G}r\"obner bases},
 Theoret.\ Comput.\ Sci.\ \textbf{134}, pp.~131--173, 1994.

\bibitem{Mora2016}
Mora,  T.,
\emph{Solving Polynomial Equation Systems IV: Volume 4, Buchberger Theory and Beyond},
Cambridge University Press, Cambridge, 2016.
 
\bibitem{Nor98}
  Nordbeck, P.,
  \emph{On some Basic Applications of Gr{\"o}bner bases in Non\nobreakdash-commutative Polynomial Rings},
  in \emph{Gröbner Bases and Applications}, pp.~463--472, Cambridge University Press, 1998.
  
\bibitem{RaabRegensburgerHosseinPoor2021}
Raab, C.~G., Regensburger, G., and Hossein Poor, J.,
 \emph{Formal proofs of operator identities by a single formal computation},
 J.\ Pure Appl.\ Algebra \textbf{225}, Article~106564, 20 pages, 2021.


\bibitem{SST00}
Saito, M., Sturmfels, B., and Takayama, N.,
\emph{Gr\"obner Deformations of Hypergeometric Differential Equations},
Algorithms Comput.\ Math.~\textbf{6}, Springer, Berlin, 2000.

\bibitem{LevandovskyySchmitz2020}
 Schmitz, L.~and Levandovskyy, V.,
 \emph{Formally verifying proofs for algebraic identities of matrices},
 in \emph{CICM 2020}, LNCS vol.~12236, pp.~222--236, 2020. 
 
\bibitem{Xiu12}
Xiu, X.,
\emph{Non-Commutative Gr\"obner Bases and Applications},
PhD thesis, University of Passau, Germany, 2012.

\end{thebibliography}
\end{document}